\documentclass[UKenglish,cleveref, autoref, thm-restate, algorithmic, algorithm]{lipics-v2021}
\usepackage[ruled, noend, algo2e]{algorithm2e}
\usepackage[noend]{algorithmic} 

\usepackage{amsmath, amsfonts, amssymb}
\usepackage{todonotes}
\usepackage{booktabs}
\usepackage{numprint}
\usepackage{changepage}
\usepackage{float}
\usepackage{tcolorbox}
\usepackage{comment}
\usepackage{enumitem}
\usepackage{tabularray}
\usepackage{hyperref}
\newcommand{\competitorCHTree}{\texttt{CH\_Tree}\xspace}
\newcommand{\competitorCQTree}{\texttt{CQ\_Tree}\xspace}
\newcommand{\competitorSimpleBTree}{\texttt{Semi-Static}\xspace}
\newcommand{\competitorDPCH}{\texttt{DPCH}\xspace}
\newcommand{\competitorFDH}{\texttt{FDH}\xspace}

\title{Engineering Fully Dynamic Convex Hulls } 
\titlerunning{Engineering Fully Dynamic Convex Hulls}

\supplementdetails[subcategory={Source Code}, cite={}, swhid={}]{Software}{https://github.com/henrixapp/PracticalConvexHulles}

\supplementdetails[subcategory={Dataset}]{Software, Data Set \& Experimental Data}{https://doi.org/10.5281/zenodo.19359873}

 \author{Ivor van der Hoog}{IT University of Copenhagen, Denmark}{ivva@itu.dk}{https://orcid.org/0009-0006-2624-0231}{}

 \author{Henrik Reinstädtler}{Heidelberg University, Germany}{henrik.reinstaedtler@informatik.uni-heidelberg.de}{https://orcid.org/0009-0003-4245-0966}{}

\author{Eva Rotenberg}{IT University of Copenhagen, Denmark}{erot@itu.dk}{https://orcid.org/0000-0001-5853-7909}{}

 \authorrunning{van der Hoog, Reinstädtler and Rotenberg}

\funding{ Ivor van der Hoog and Eva Rotenberg are grateful to the VILLUM Foundation for supporting this research via Eva Rotenberg’s Young Investigator grant (VIL37507) “Efficient Recomputations for Changeful Problems”. Henrik Reinstädtler acknowledges support by DFG grant SCHU 2567/8-1.}
\begin{document}

\ccsdesc[100]{Theory of computation~Algorithm engineering}
\ccsdesc[100]{Theory of computation~Computational Geometry} 

\keywords{Convex hulls, fully-dynamic data structures, robustness}

\hideLIPIcs

\maketitle

\begin{abstract}
We present a new fully dynamic algorithm for maintaining convex hulls under insertions and deletions while supporting geometric queries.
Our approach combines the logarithmic method with a deletion-only convex hull data structure, achieving amortised update times of $O(\log n \log \log n)$ and query times of $O(\log^2 n)$.
We provide a robust and non-trivial implementation that supports point-location queries, a challenging and non-decomposable class of convex hull queries.

We evaluate our implementation against the state of the art, including a new naive baseline that rebuilds the convex hull whenever an update affects it.  On hulls that include polynomially many data points (e.g. $\Theta(n^\varepsilon)$ for some $\varepsilon$), such as the ones that often occur in practice, our method outperforms all other techniques.
Update-heavy workloads strongly favour our approach, which is in line with our theoretical guarantees. Yet, our method remains competitive all the way down to when the update to query ratio is $1$ to $10$.  

Experiments on real-world data sets furthermore reveal that existing fully dynamic techniques suffer from significant robustness issues. In contrast, our implementation remains stable across all tested inputs.
\end{abstract}

\setcounter{page}{0}

\newpage
\section{Introduction}

Let $P$ be a planar point set of size $n$. The convex hull of $P$ is the smallest~convex area enclosing $P$.
In analysis of spatial data, computing the convex hull is a fundamental task. 
The convex hull moreover frequently serves as a computational stepping stone towards further data processing. Seen from the \emph{static} perspective, the problem of computing the convex hull has received much attention.
In this paper, we study the \emph{dynamic} setting, where data points may arbitrarily be added, deleted, or altered (which can be handled by a deletion plus an insertion).
Convex hulls have many applications for evolving or changing data sets, illustrating the need for an efficient, practical, and publicly available dynamic convex hull implementation.
Convex hulls are among the most~studied objects in computational geometry~\cite{Avis1995How, brodal2002dynamic}, with applications in clustering~\cite{gao2017trajectory,liparulo2015fuzzy,sander1998density}, shape analysis~\cite{Furukawa1995, Wang2007Supervised}, data pruning~\cite{giorginis2022fast, khosravani2016convex, margineantu1997pruning, ostrouchov2005fastmap}, selection queries~\cite{ihm2014approximate, mouratidis2017geometric}, and road network analysis~\cite{gao2017trajectory,liu2019fast, yan2011efficient}.

\subparagraph{Convex hulls in literature.}
We denote by $CH(P)$ the convex hull edges in cyclic order and $h  := |CH(P)|$. 
For dynamic convex hulls, we make a distinction between \emph{explicit} and \emph{implicit} convex hulls. 
An \emph{explicit} convex hull data structure stores, at all times, the ordered sequence $CH(P)$ --- for example, as a balanced binary tree. 
The current state-of-the-art algorithm for maintaining an explicit fully dynamic hull is still the classical algorithm by Overmars and van Leeuwen~\cite{overmars1980dynamically}. Their \texttt{CQ-Tree} can explicitly maintain $CH(P)$ as a balanced tree of depth $O(\log h)$ subject to point insertions and deletions in $O(\log^2 n)$ time. This remains the most efficient \emph{explicit} theoretical method to date.
A small adaption of this technique called an \texttt{CH-Tree}~\cite{Gaede2024Simple} offers the same update time, worse $O(\log n)$ query time, but better practical performance. 
Other explicit works can achieve better runtime only by restricting the problem statement. 
Brewer et al.~\cite{brewer2024Dynamic} restrict $P$ to be an ordered set that form the vertices of a simple path. Specifically, the updates in \cite{brewer2024Dynamic} allow for appending points such that the resulting path is without any crossings, and the reverse pop operation that removes the previously appended points.
When explicitly maintaining the convex hull in a balanced tree, their updates require worst-case $O(\log n)$ time.
Alternatively, their updates take worst-case $O(1)$ time and they can make the hull explicit at will using $O(\log n)$ time. 
Wang~\cite{Wang2023Dynamic} offers the same trade-off, but requiring only $O(\log h)$ instead of $O(\log n)$ time, in the restricted update model where $P$ is required to remain $x$-monotone under append and dequeue updates.

\subparagraph{Implicit convex hulls.}
An \emph{implicit} convex hull data structure stores $P$ in a data structure subject to convex hull queries. The typical convex hull queries considered are:~\cite{chan2011three}:
\begin{enumerate}[noitemsep, nolistsep]
    \item \label{it:extreme}
    Finding the extreme point of $P$ in a query direction, 
    \item Deciding whether a line intersects $CH(P)$, 
    \item Finding the two hull vertices tangent to a query, \label{it:tangent}
    \item Deciding whether a query point $q$ lies in $CH(P)$, \label{it:inside}
    \item Finding the intersection with a query line, \label{it:intersectline}
    \item Finding the intersection between two hulls. \label{it:hullintersect}
\end{enumerate}

Of these queries, the first three are \emph{decomposable}. Formally, consider a partition of $P$ into point sets $P_1, \ldots, P_k$.
If we have, for the first three queries, the query answer for each part $P_i$ then we can compute the query answer on $P$ in $O(k)$ time. 
Decomposable queries nicely allow the application of divide-and-conquer approaches. 
Friedman et al.~\cite{Friedman1996Efficiently} considered implicit structures first. Their data structure was geared towards decomposable queries. 
They restrict $P$ to a simple path and support appending and depending points in amortised $O(\log n)$ and $O(1)$ time, respectively, and answer decomposable queries in $O(\log n)$ time. 

Chan's~\cite{chan2001dynamic}  structure has $O(\log^{1+\varepsilon} n)$ amortised update time and answers decomposable queries in $O(\log^{3/2} n)$ time. 
Chan later gave an improved algorithm with an expected $O(\log^{1+\varepsilon} n)$ query time for decomposable queries~\cite{chan2011three}.
Brodal and Jacob~\cite{Brodal2000Dynamic} support amortised $O(\log n \log \log n)$ updates and answer decomposable queries in $O(\log n)$ query time. 
The state-of-the-art by Brodal and Jacob~\cite{brodal2002dynamic} supports updates in amortised $O(\log n)$, decomposable queries in $O(\log n)$, and non-decomposable queries in $O(\log^{3/2} n)$ time.

\subparagraph{Insertion- or deletion-only structures.}
Preparata~\cite{Preparata1979} introduced the first explicit insertion-only convex hull algorithm, achieving $O(\log n)$ amortised update time while supporting queries in $O(\log h)$ time. A simpler folklore adaptation of Graham’s scan~\cite{graham1972efficient}, combined with a balanced binary tree, attains $O(\log h)$ update time and supports queries in $O(\log n)$ time.

Van der Hoog, Reinstädtler, and Rotenberg~\cite{vanderhoogetal2026} engineer efficient insertion-only solutions. They compare the two techniques above with a new approach based on the logarithmic method of Overmars. This latter technique supports updates in amortised $O(\log n \log \log n)$ time and queries in $O(\log^2 n)$ time, but its vector-based approach is often faster in practice.

Chazelle~\cite{Chazelle1985Convex}, and Hershberger and Suri~\cite{hershberger1992applications} maintain for a deletion-only sequence the convex hull in a doubly linked list. Deleting a size-$n$ point set takes $O(n \log n)$ total time. 

\subparagraph{Implementations of convex hull algorithms.}
The number of explicit implementations of dynamic convex hull data structures remains limited, though it is slowly growing. Implicit methods~\cite{brodal2002dynamic, chan2011three, chan2001dynamic} have not been implemented to date, largely because they rely on intricate cutting-based techniques that are challenging to realise in practice.

Perhaps surprisingly, standard geometric libraries such as the BOOST Geometry library~\cite{boost-geometry} and the Computational Geometry Algorithms Library (CGAL)~\cite{cgal} do not provide implementations of dynamic convex hull data structures. While CGAL does support dynamic Delaunay triangulations --- and the convex hull is a subgraph of this graph --- such triangulations incur poor linear-time update costs when used for hull maintenance.

Outside of standard libraries, we are aware of four available implementations:
\begin{itemize}
    \item Chi et al.~\cite{chi2013distribution} provide a Java implementation of the Overmars--van Leeuwen algorithm; however, it lacks query support, robustness, and efficiency~\cite{Gaede2024Simple}.
    \item Gæde et al.~\cite{Gaede2024Simple} present robust implementations of a \texttt{CQ-Tree} (which is the classical algorithm from \cite{overmars1980dynamically}) and derive from it their \texttt{CH-Tree} data structure. 
    \item Shirgure~\cite{sumeet_shirgure_2023_8396184} provides another implementation of the \texttt{CQ-Tree} from~\cite{overmars1980dynamically}.
    \item Van der Hoog et al.~\cite{vanderhoogetal2026} implement and evaluate several insertion-only techniques.
\end{itemize}

\subparagraph{Contribution.}
In the insertion-only setting, the implementation of Van der Hoog, Reinstädtler, and Rotenberg~\cite{vanderhoogetal2026} achieves significantly faster update times by avoiding tree rotations and pointer traversals. Their method, based on the logarithmic method of Overmars~\cite{overmars1983design}, partitions the input point set $P$ into buckets $\{B_i\}$ and stores $CH(B_i)$ in a contiguous vector.  Insertions merge buckets, invoking a linear-time static convex hull construction algorithm for inputs sorted by $x$-coordinate. While this yields low storage overhead and favourable memory access, storing $CH(B_i)$ in a vector cannot be made fully dynamic. 

In this paper, we revisit and generalise this approach. Instead of storing the convex hull of each bucket $B_i$ in a vector, we adapt the decremental convex hull technique from~\cite{Chazelle1985Convex, hershberger1992applications} to maintain a decremental hull within each bucket. This raises several challenges. First, to match $O(n \log n \log \log n)$ update time from~\cite{vanderhoogetal2026} as possible, we need to develop a linear-time construction algorithm of the data structures in~\cite{Chazelle1985Convex, hershberger1992applications} when the input is sorted.

Second, the logarithmic method is not designed for deletions. 
We show that the typical \emph{tombstoning} method used to proxy deletions does not work for convex hulls. Instead, we must rely on deletion-only data structures in each node which requires new insights and analysis. 

Third, the algorithms in~\cite{Chazelle1985Convex,hershberger1992applications} are not decremental in the traditional sense: they only guarantee a total running time of $O(n \log n)$ over $n$ deletions. This property complicates its integration into the logarithmic method. Moreover, the original algorithm maintains the convex hull only as a doubly linked list. 
This does not support efficient queries, as traversing a linked list takes linear time. 
Finally, since our approach partitions $P$ along buckets, we furthermore encounter difficulties when supporting non-decomposable queries. 

We compare our implementation  to all publicly available fully dynamic convex hulls. Our experimental study spans a wide range of synthetic and real-world data sets. Handling real-world data introduces further challenges, as geometric algorithms typically assume general position. Yet, practical data contains many points with identical $x$- or $y$-coordinates. We adapt our algorithm to handle such degeneracies robustly and demonstrate that existing implementations fail to do so.
Our proposed technique achieves substantially faster update times, at the cost of increased query times. For workloads with a mix of updates and queries, our approach is preferable, remaining competitive until queries account for more than $90\%$ of the workload. Overall, we provide new theoretical insights together with an efficient and robust implementation for fully dynamic convex hull maintenance under mixed workloads.

\section{Preliminaries}

Let $P$ be a set of $n$ points in the plane $\mathbb{R}^2$. 
The convex hull $CH(P)$ is the set of edges bounding the smallest convex region containing $P$. 
The \emph{upper convex hull} $CH^+(P)$ is defined as the convex hull of $P \cup \{(0,-\infty)\}$
and the \emph{lower convex hull} $CH^-(P)$ is defined analogously using $P \cup \{(0,\infty)\}$.
Then $CH(P)$ can be found by intersecting $CH^+(P)$ and $CH^-(P)$. 
As upper and lower hulls are symmetric, we focus on upper hulls.
Let $A$ and $B$ be point sets separated by a vertical line.
The upper hull $CH^+(A\cup B)$ consists of a prefix of $CH^+(A)$ and a suffix of $CH^+(B)$, joined by a single edge called a \emph{bridge}.
Formally, the bridge is the minimal segment $\overline{ab}$ with $(a,b)\in A\times B$ such that $A\cup B$ lies below its supporting line. 

\subparagraph{CH-Trees.}
The baseline data structure used by all fully dynamic convex hull approaches is  a \texttt{CH-Tree}. 
This is a balanced binary tree $T(P)$ that stores $P$ in its leaves, sorted by $x$-coordinate. 
For any node $x \in T(P)$, $\pi(x)$ denotes the points stored in the subtree under $x$. 
A \texttt{CH-Tree} stores for each node $x$ with children $u$ and $v$ the unique bridge between $CH^+(\pi(u))$ and $CH^+(\pi(v))$.  A  \texttt{CH-Tree} can be dynamically maintained in $3$ ways: classically it is augmented to a \textbf{CQ-Tree} which maintains in each node of the \texttt{CH-Tree} another balanced tree. Deletion-only, \texttt{CH-Trees} can be maintained in a simpler manner~\cite{Chazelle1985Convex, hershberger1992applications}. Recently, Gaede et al.~\cite{Gaede2024Simple} showed an algorithm to maintain a \texttt{CH-Tree} as a standalone data structure.

\subparagraph{CQ-Trees.}
Overmars and van Leeuwen~\cite{overmars1980dynamically} introduced the \texttt{CQ-Tree}, which is a \texttt{CH-Tree} that stores in each node $x \in T(P)$ a \emph{concatenable queue}. 
A concatenable queue $\mathbb{E}(x)$ stores the edges of $CH^+(\pi(x))$ that are not in $CH^+(\pi(y))$ where $y$ is an ancestor of $x$. 
They prove that these edges are a contiguous subsequence of $CH^+(\pi(x))$ and a concatenable queue stores this sequence in a balanced binary tree. 
Note that for the root $r$ of $T(P)$, $\mathbb{E}(x)$ stores $CH^+(P)$ as a balanced binary tree.
For any node $x \in T(P)$, they can construct $CH^+(\pi(x))$ top-down by iteratively splitting the convex hull of the parent node at the bridge, and merging it with the concatenable queue of the child node. 
Since splitting and merging balanced trees takes $O(\log n)$ time per operation, this takes $O(\log^2 n)$ total time. This procedure is then also used to maintain the underlying \texttt{CH-Tree}:
any point update updates a root-to-leaf path in $T(P)$. Along this path, they obtain for each node  $x$ the hulls $CH^+(\pi(u))$ and $CH^+(\pi(v))$ by splitting and merging concatenable queues. 
Given these hulls, they update the bridge of $x$ in logarithmic time. 
Since the root of \texttt{CQ-Tree} stores the convex hull in a balanced tree, a \texttt{CQ-Tree} on $P$ (and $P$ mirrored around the $x$-axis) can answer convex hull queries in $O(\log h)$ time where $h$ is the number of convex hull edges. 
Implementations of \texttt{CQ-Trees} are relatively inefficient since tree splits, merges, and rotations are quite expensive.

\subparagraph{Decremental CH-Trees.}
The decremental convex hull~\cite{Chazelle1985Convex, hershberger1992applications} adapts the approach from~\cite{overmars1980dynamically} for decremental convex hulls. 
It stores $\mathbb{E}(x)$ as a doubly linked list instead.
This has the benefit that splits and merges become constant-time but has a downside that searching over a hull takes linear time.
Let $n$ be the original size of $P$, their algorithm does not rotate $T(P)$ and instead only deletes nodes, maintaining a depth of $O(\log n)$ where $n$ is the original input size. 
Secondly, the algorithm can delete \emph{all} points of $P$ in $O(n \log n)$ total time. These are weaker guarantees than typical deletion-only analysis provides: where $k$ deletions can be supported in $O(k \log n)$ total time.
Under these conditions, \cite{Chazelle1985Convex, hershberger1992applications}  observe that 
if a point $p$ appears on $CH^+(\pi(x))$ for some node $x \in T(P)$ then it will remain on this hull until it is deleted. 
Upon deleting a point $p$, $T(P)$ is updated bottom-up. 
If for a node $x \in T(P)$, $p$ does not lie on $CH^+(\pi(x))$ then nothing changes.
Otherwise, the algorithm performs a linear scan to fill the 'gap' left by $p$ using the updated $\mathbb{E}(\pi(u))$ and $\mathbb{E}(\pi(v))$ for $x$'s children $u$ and $v$. 
This scan takes time linear in the number of points that fill the gap. 
Since each point found is a new point that will remain on $CH^+(\pi(x))$ since it is deleted, and each point appears in $\pi(x)$ for $O(\log n)$ nodes $x \in T(P)$, deleting all of $P$ takes $O(n \log n)$ total time.

\subparagraph{Dynamic CH-Trees.}
Gæde et al.~\cite{Gaede2024Simple} observe the practical cost of \texttt{CQ-Trees} and that a \texttt{CH-Tree} can be maintained with $O(\log^2 n)$ update time without storing any information besides the bridges.
For queries, they note again that using only bridges suffices. 
For example, a point $p$ is contained in $CH^+(P)$ if and only if there exists a bridge $e$ in $T(P)$ where $p$ lies below $e$. The result is a fully dynamic algorithm to compute a \texttt{CH-Tree} in $O(\log^2 n)$ time that can answer convex hull queries in $O(\log n)$ time (as opposed to the $O(\log h)$ time of \texttt{CQ-trees}). In practice, \texttt{CH-Trees} are more efficient as they require fewer tree operations. 

\subparagraph{Insertion-only.}
Van der Hoog et al.~\cite{vanderhoogetal2026} apply the logarithmic method to the insertion-only convex hull problem.
When the input point set $P$ is sorted by increasing $x$-coordinate, the upper convex hull $CH^+(P)$ can be computed in linear time using Graham’s scan~\cite{graham1972efficient}.
The logarithmic method maintains buckets $B_i$ of sizes $2^i$, for $i \in O(\log n)$.
Each bucket is either empty or completely filled, and all filled buckets store their points in sorted order.
Upon inserting a new point, the algorithm identifies the largest index $j$ such that all buckets $B_i$ with $i \le j$ are full.
These buckets are then cleared, and their content plus the new point are merged into bucket $B_{j+1}$.
A $O(\log n)$-way sorting merge takes $O(|B_{j+1}|\log\log n)$ time.
Once merged, the algorithm recomputes the upper convex hull $CH^+(B_{j+1})$ in linear time using Graham’s scan.
In this way, the point set $P$ is maintained as a disjoint union of $O(\log n)$ buckets, each with its own upper convex hull.
As each point is involved in at most $O(\log n)$ merges over the course of $n$ insertions, the amortised update time is $O(\log n \log\log n)$.

This method immediately answers any decomposable query at an $O(\log n)$ overhead by answering the query for each bucket. Non-decomposable queries are considerably more difficult and the focal point in~\cite{vanderhoogetal2026} is deciding whether a point $q$ lies in $CH(P)$.

\section{Theory and implementation}
\label{sec:theory}

We wish to dynamically maintain a point set $P$ to support queries on the (upper) convex hull.
Firstly we note that geometric algorithms typically assume that $P$ lies in a general position where no points share $x$ or $y$ coordinates. 
We note that for our real-world input instances, this assumption is very frequently violated and so we must deviate from this standard. 
Most algorithmic difficulties arise from having a shared $x$-coordinate and for queries on the upper convex hull it always suffices to keep for each $x$-coordinate only the point with maximal $y$-coordinate. 
Thus, our first data structure is a dictionary on $x$-coordinates where for each $x \in \mathbb{R}$ such that there exists a point $p \in P$ with $x = p.x$, we store a priority queue of all points in $P$ that share this $x$-coordinate. 
The value of each point in the priority queue is the $y$-coordinate and we only store the point $p_x^*$ with maximal $y$-coordinate in the remaining data structure.
If we insert a new point $p$ with the same $x$-coordinate but greater $y$-coordinate, we delete $p_x^*$ from the remaining data structure and insert $p$ instead. 
As a consequence, we assume that $P$ denotes the `remaining' point set where each point has a unique $x$-coordinate.

\subparagraph{The logarithmic method.}
The logarithmic method by Overmars~\cite{overmars1983design} is a technique for making any static data structure with decomposable queries an insertion-only data structure, with a factor $O(\log n)$ overhead in update and query time. 
Its classical application is in rank queries, where we want to maintain a set of values $V$ and ask for a query value $q$ how many $v \in V$ are smaller than $q$. 
Rather than maintaining $V$ in a balanced tree, one can instead partition $V$ across buckets $B_i$ where each bucket $B_i$ has size $2^i$ and is either empty or full.
Full buckets store their content in sorted order. 
Inserting a new value $v'$ into $V$, identifies the maximum integer $j$ such that $B_1, \ldots, B_j$ are all full and invokes $\texttt{Merge}(j)$. 
This procedure clears all these buckets and puts their content (and $v'$) into $B_{j+1}$ via a $j$-way merge which takes $O(|B_{j+1}| \log j)$ time. 
This structure requires linear space. 
Since after $n$ insertions, each point participates in at most $O(\log n)$ merges the total running time after $n$ insertions is $O(n \log n \log \log n)$.
Any decomposable query can be answered by querying each bucket $B_i$ separately. E.g., the rank of a query $q$ equals the sum of the ranks in each bucket. 

\subparagraph{Complications for convex hulls.}
A convex hull can be constructed in linear time, if the input is sorted by $x$-coordinate.
Van der Hoog, Reinstädtler, and Rotenberg~\cite{vanderhoogetal2026} use this technique to store $P$ subject to convex hull queries using the logarithmic method.
They augment the procedure  $\texttt{Merge}(j)$ such that after bucket $B_{j+1}$ is filled, one can invoke Graham's scan~\cite{graham1972efficient} to compute $CH^+(B_{j+1})$ in linear time which incurs no asymptotic overhead. 
Decomposable queries can then simply be handled by querying all $O(\log n)$ separately but many convex hull queries are non-decomposable. Van der Hoog, Reinstädtler, and Rotenberg~\cite{vanderhoogetal2026} focus on point location (although all other queries can be handled by adapting this procedure). 
They observe the following: if, for a bucket $B_i$, a query point $q$ lies inside $CH^+(B_j)$ then $q$ lies inside $CH^+(P)$. The difficult case arises when $q$ lies outside of $CH^+(B_i)$ for all $i$, as it may still be the case that $q$ lies inside $CH^+(P)$. 
They compute for each bucket $B_i$ the \emph{tangents}, which are the edges $(a_i, q)$ and $(q, b_i)$ of $CH^+(B_i \cup \{ q \})$, in logarithmic time.
Let $X = \{ a_i \} \cup \{ b_i \}$. They prove that $q$ is inside $CH^+(P)$ if and only if $q \in CH^+(X)$. Since $X$ has at most $O(\log n)$ points, the running time of this procedure is dominated by $O(\log^2 n)$. 

We will also apply the logarithmic method, and thus require the same consideration.
However, more problems arise  with deletions. Classically, deletions are handled through \emph{tombstoning}. Instead of deleting a value $v$ from $V$ one inserts a \emph{tombstone} $v^{-1}$. During $\texttt{Merge}(j)$ we can detect that $B_{j+1}$ contains both $v$ and $v^{-1}$ and delete both values permanently. 

Under tombstoning, buckets contain both values and tombstones, which are sorted separately. 
Tombstones are useful for queries that return cumulative info over all the input.
For example for rank queries, one can compute for a query $q$, for each bucket $B_i$, the number of values and tombstones in $B_i$ that precede $q$.
The rank of $q$ is the total number of values in the data structure that precede $q$, minus the number of tombstones that precede $q$ and rank queries can thus still be answered in $O(\log^2 n)$ time. 
For convex hulls, this entire technique breaks down as the answer is a Boolean and not a sum: 
Consider the naive scenario where there are two buckets $B_2$ and $B_6$ where $B_2$ contains only tombstones of values in $B_6$. 
It may be that a query $q$ lies under some edge $(a, b)$ of $CH^+(B_6)$, but that $B_2$ contains the tombstones $a^{-1}$ and $b^{-1}$.
When removing these points from $CH^+(B_6)$ it may be that $q$ lies above the convex hull, or below some edge whose endpoints are not even on $CH^+(B_6)$ and any algorithm has no way of knowing which is which, without performing actual deletions. 

\subsection{Our data structure}\label{sec:our:datastructure}
Since we cannot implement tombstoning, we are required to adapt the logarithmic method and use deletion-only data structures instead. At all times, we partition $P$ along buckets $B_i$ of size $2^i$. Contrary to the standard logarithmic method, we do not require that all buckets are either empty or full. 
For each bucket $B_i$, we store the deletion-only convex hull data structure from~\cite{Chazelle1985Convex, hershberger1992applications} denoted by $T(B_i)$. 
The trees $T(B_i)$ are never rotated. This allows us for an efficient implement the tree as a vector. We will ensure this property for all trees used throughout this paper to increase efficiency. We first adapt the \texttt{Merge} procedure:

\subparagraph{Defining $\texttt{Merge}(j)$.}
We develop an algorithm to construct $T(P')$ in linear time given a sorted point set $P'$.
Our merging algorithm first does an $O(j)$-way merge to obtain a point set sorted by $x$-coordinate and then constructs the deletion-only data structure in linear time. 
Contrary to the classical method, we will no longer have the guarantee that the bucket $B_{j+1}$ is empty when we invoke this procedure. 
Instead, we will keep the invariant that $B_{j+1}$ is at most quarter-full.
On a high level, we clear the content of all buckets $B_1, \ldots, B_{j+1}$.
We then do a $(j+1)$-way merge. For efficiency, we implement this merge in-place using a \emph{loser tree}.
A loser tree is a balanced tree with $(j+1)$ leaves, each leaf $i$ storing the remaining minimum value in $B_i$. 
Each node stores the smaller value of its children and so the root stores the remaining minimum value in $B_1, \ldots, B_{j+1}$. Iteratively remove this minimum value, updating a path of length $O(\log j)$ in this tree, merging all values from $B_1, \ldots, B_{j+1}$ in sorted order. We then find two buckets $B_a$, $B_b$ with $a, b \in [1, j+1]$ such that all content fits in these buckets and $B_a$ and $B_b$ and we split the merged content into these buckets.

\subparagraph{Dynamic maintenance.}
For deleting some $p$ from $P$, locate the unique bucket $B_i$ and delete $p$ from $T(B_i)$. 
At all times, we count the number of remaining elements in $T(B_i)$.
If this number drops below $2^{i - 2}$ then we trigger $\texttt{Merge}(i - 1)$.
When inserting a point $p$ into $P$, we identify the maximum integer $j$ such that $B_1, \ldots, B_j$ are non-empty and trigger $\texttt{Merge(j)}$.
Proving that this takes $O(n \log n \log \log n )$ amortised time requires some careful analysis:

\begin{theorem}
    The above procedure takes amortised  $O(n \log n \log \log n )$  time.
\end{theorem}

\begin{proof}
    Our algorithm incurs running time in two ways: through deletions in $T(B_i)$ for some bucket $B_i$ and through $\texttt{Merge}(j)$ for indices $j$. 

    We first analyse the running time by $\texttt{Merge}(j)$ for $j \in O(\log n)$. 
    Note that $\texttt{Merge}(j)$ takes $O(2^j \log \log n)$ time.
    We make a case distinction for how this event gets triggered. 
    If this event is triggered by deletions, then the bucket $B_{j+1}$ dropped below size $2^{j-1}$.
    Since the bucket $B_{j+1}$ as size at least $2^{j}$ when it is initialised, we can charge $O(\log \log n)$ time to each deleted element to pay for this merge. Since these elements are deleted, they can never be charged again. 
    Insertions are more complicated as $\texttt{Merge}(j)$ can populate buckets $B_a$ and $B_b$ and it may be that $a, b < j + 1$. 
    The key observation is that if we invoke  $\texttt{Merge}(j)$  with an insertion then $B_{j+1}$ is empty.  If this procedure \emph{does not} populate $B_{j+1}$ then there are at most $2^j$ elements in $B_1, \ldots, B_j$. 
    The buckets $B_1, \ldots, B_j$ fit, together $2^{j+1} - 1$ elements.
    We greedily fill at most two buckets which are both at least half-full. 
    It follows that this procedure leaves some buckets in $B_1, \ldots, B_j$  empty and that the number of empty slots in these buckets is at least $\frac{1}{4}2^j$. 
    As a consequence, we require at least  $\frac{1}{4}2^j$ new insertions before we trigger $\texttt{Merge}(j)$ again and we charge these insertions $O(\log \log n)$. 
    In the resulting charging scheme, we only pay for the first time $\texttt{Merge}(j)$ gets invoked for all $j \in O(\log n)$ and thus the amortised running time after $n$ updates is $O(n \log n \log \log n)$.

    The data structure $T(B_i)$ can support a sequence of $|B_i|$ deletions in $O(|B_i| \log |B_i|)$ time.
    However, this running time can be front-loaded: the first $O(\log |B_i)$ deletions in $T(B_i)$ may all take $O(|B_i|)$ time. The classical argument guarantees that the remaining deletions take constant time and charges the total running time to these non-costly deletions, ensuring amortised $O(\log |B_i|)$ time deletions.
    However in our case, the procedure $\texttt{Merge}(j)$ for $j \geq i-1$ may merge the content of $B_i$ into some other bucket --- where afterwards this front-loaded behaviour can occur again. We can incorporate this into the analysis by identifying a geometric sum to erase this charge.
    Formally, each time we perform a deletion on a bucket $B_i$ we add the running time that the deletion takes as a \emph{charge}.
    For each deletion, we also reduce the charge by some constant. 
    If we trigger $\texttt{Merge}(j)$, then we charge all the built-up charge in buckets $B_i$ for $i \leq j$ to what triggered the merge. 
    Due to bucket sizes doubling as we increase $i$, there can be at most $O(2^j \log (2^j))$ total charge. The first time we invoke $\texttt{Merge}(j)$, we pay that charge to the overall running time. 
    We already showed that for each subsequent call to  $\texttt{Merge}(j)$ is called, $\Omega(2^j)$ unique updates can me charged for this call. 
    It follows that the overall running time is at most:
    $
    \sum\limits_{j = 1}^{\log n} O(2^j \log (2^j)) \subset O(n \log n). \qedhere
    $
\end{proof}

\subparagraph{Further challenges.}
We cannot present all implementation details in the main body, and will refer to the provided code repository for details. However, we briefly highlight the challenges that we encountered and that our codebase handles in a non-trivial manner.

Firstly, we require as a subroutine that for any sorted point set $P$,  we can construct the data structure $T(P)$ from~\cite{Chazelle1985Convex, hershberger1992applications} in linear time. It may not be surprising that this is possible, yet there are many details to implementing this. 
$T(P)$ is a balanced tree that sorts $P$ on $X$. For each node $x \in T(P)$ it requires us to identify the contiguous subsequence of $CH^+(\pi(x))$ that is not part of $CH^+(\pi(y))$ where $y$ is the parent of $x$ --- all in total linear time. 
This requires a careful bottom-up construction where for each depth $d$ we first construct for all nodes $y$ at depth $d$ the hull $CH^+(\pi(y))$ from all $CH^+(\pi(x))$ for nodes $x$ at depth $d+1$. 

Secondly, we implement point location queries which are non-decomposable. In each bucket, we store a \texttt{CH-Tree} $T(B_i)$. 
 Gæde et al.~\cite{Gaede2024Simple} showed that a \texttt{CH-Tree} suffices to answer standard convex hull queries. However, we have $P$ partitioned across buckets and thus require the query approach by~\cite{vanderhoogetal2026}. This approach requires an additional query which for every point $x$ not under $CH^+(B_i)$ finds its two \emph{tangents}.
  Gæde et al.~\cite{Gaede2024Simple} can find these tangents in $O(\log^2 n)$ time by inserting the query point, which is not fast enough. Finding these tangents in $O(\log n)$ time requires considerable care. 
The result is an algorithm with amortised $O(n \log n \log \log n)$ update time that facilitates queries in $O(\log^2 n)$ time. 

\section{Experiments}
In this section we evaluate our algorithm from Section~\ref{sec:theory} and compare it to state-of-the-art methods.
We benchmark  insertion, deletion, and query time and check for query correctness.
We evaluate 5 fully-dynamic algorithms, all implemented in C++:
\begin{itemize}
\item \competitorSimpleBTree, a naive algorithm that stores the convex hull in a vector and the input points in sorted order. Each update, we first test in logarithmic time if the update changes the convex hull. If so, we recompute the convex hull in linear time using Graham scan. 
    \item \competitorCHTree and \competitorCQTree from \cite{Gaede2024Simple}
    \item \competitorDPCH from \cite{sumeet_shirgure_2023_8396184}
    \item \competitorFDH the implementation of Section~\ref{sec:our:datastructure} with initial bucket size of 32 and 1024. %
\end{itemize}
\subparagraph*{Data sources}
Following previous work on static and dynamic convex hull computation \cite{Cadenas2019Preprocessing,gamby2018convex,Gaede2024Simple,vanderhoogetal2026} we test on data verifying real-world, synthetic  and scaling data.

\emph{Real-World Data.}
We use three data sets proposed by \cite{vanderhoogetal2026} for insertion-only experiments  and enhance them with deletions.
We delete 50\% of overall points and mix query, insertion and deletion randomly, while making sure that deleted points have been inserted prior.
The first data set is the 3D \textbf{Mammal} set, also used by \cite{Cadenas2019Preprocessing,Mei2016Gang}, containing 3D scans of mammalian bodies. We test every 2D-projection combination.
The \textbf{Tiger} set, used for range queries~\cite{Govindarajan2003CRB,Danzhou2002Efficient}, contains geographical features from the US Census Bureau~\cite{tiger2006se}.
Lastly, we use the \textbf{Cluster} set~\cite{ClusteringDatasets}, which is widely used in benchmarking clusters~\cite{Chang2008Robust, Fu2007FLAME, Gionis2007Clustering, Jain2005Data, Zahn1971Graph}.

\emph{Synthetic Data.}
For synthetic data we use the generators from \cite{gamby2018convex,Gaede2024Simple} with the parameters set in \cite{vanderhoogetal2026}.
The four generators follow simple geometric patterns and are listed in Table~\ref{tab:generators}.
Similarly to the real-world data set, we test two rounds of the schema: insertion, query, and deletion.
In the first round, we insert one million points, run one million queries, and delete half of the points.
In the second round, we add one million additional points, run one million queries, and delete a quarter million points.
This ensures that we a fully dynamic scenario. For every generator we have ten instances with different seeds.

\emph{Scaling Data.}
Similar to~\cite{vanderhoogetal2026} we benchmark our algorithms on input scaling from  $2^{20}$ to $2^{24}$ with the four synthetic generators as input.
For every generator, we use 5 instances with 50~\% deletions and 5 instances with 25~\% deletions. Between insertion and deletion we run equally many queries as the input size.
Every instance uses a different seed.

\emph{Query Correctness}
We checked whether the reported size of the convex hull and the number of true queries match.
For synthetic data, for which points are in general position, all methods report the same numbers.
On the real world datasets, which contain partially duplicate or fully duplicate points, the number of true answers varies, while \competitorCQTree crashes on the sequences generated from the \textbf{Cluster} data set.
Only \competitorSimpleBTree and our methods \competitorFDH agree on the number of true queries.
Moreover, we were able to engineer a point-sequence that lets \competitorCHTree loop indefinitely.
We conclude that the other methods are not capable of maintaining convex hulls in non-general positions.

\begin{table}[H]
    \centering
    \begin{tabular}{>{\bf}llr}
    \toprule
    Distribution & \bf Description & \bf \# Hull-points\\\midrule
       Box  & Uniformly sampled from a square (side-length=1000)  & $\Theta(\log{n})$\\
       Bell & 2D-Gaussian & $\Theta(\log{n})$\\
       Disk & Uniformly sampled from disk (radius=1000) & $\Theta(n^{1/3})$\\
    Circle & Uniformly sampled from circle (radius=1000) & $\Theta(n)$\\
         \bottomrule
    \end{tabular}
    \caption{Distributions used for generating point clouds, parameters follow \cite{vanderhoogetal2026}.}
    \label{tab:generators}
\end{table}

\emph{Methodology and organisation.}
We use two systems equipped with a 3.10Ghz Xeon w5-3435X processor (45MB cache) and 128 GB of RAM and schedule up to 4 experiments in parallel.
The memory limit per experiment is 90 GB.
We only compare results obtained from the same machine.
Every experiment with real-world data is repeated 5 times and 3 times with synthetic data.
Every reported metric is the mean of these running times.
Section~\ref{sec:realworld} contains the discussion of the real-world data sets. Section~\ref{sec:synth} discusses the results obtained from synthetic benchmarks, while Appendix~\ref{sec:scaling} discusses scaling results.

\subsection{Real world Data}\label{sec:realworld} 

We first discuss the experimental results on the real-world data sets.
Our original intention was to use these data sets to compare algorithmic performance.
Instead, the experiments reveal that they primarily serve as a robustness check, as there are many points with duplicate $x$- and $y$-coordinates.
\competitorSimpleBTree, which is essentially a static algorithm and therefore comparatively easy to make robust, serves as the ground truth.
We observe the following:

\competitorDPCH segfaults on nearly half of the queries on the \textbf{Mammals} data set and crashes outright on \textbf{Cluster}.
It furthermore times out on \textbf{Tiger}.
The same holds for \competitorCQTree, except that it does not time out on \textbf{Tiger} but instead crashes.
Finally, \competitorCQTree leaks and runs out of memory on \textbf{Tiger}.
In contrast, our approach is robust across all data sets.

\subparagraph{Comparing update times.}
As a consequence of the robustness issues, our ability to compare these algorithms on this data set is limited.
We nevertheless highlight a few trends that also persist on synthetic data (see Figure~\ref{fig:realworld:update}).
On \textbf{Clustering} and \textbf{Tiger}  (and, more generally, on data sets with sparse convex hulls), the naive \competitorSimpleBTree algorithm exhibits very strong performance.
If the convex hull is sparse, many point deletions remove points that do not lie on the convex hull, while many point insertions place points deep inside the hull.
In both cases, \competitorSimpleBTree detects this situation in logarithmic time and terminates early.
One could therefore argue that, for many real-world scenarios, a fully dynamic data structure is not strictly necessary.
If we discount this phenomenon, and further exclude \competitorCQTree and \competitorFDH due to their frequent crashes, we can meaningfully compare our \competitorFDH algorithm with \competitorCHTree.
The performance difference between the two \competitorFDH implementations is marginal, with a bucket size of 1024 being slightly more efficient.
Overall, our approach is an order of magnitude faster than \competitorCHTree in terms of update time.
This trend also persists on synthetic data.

\begin{figure}[H]
    \centering
      \includegraphics[]{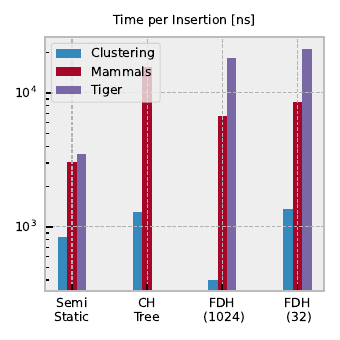}
    \includegraphics[]{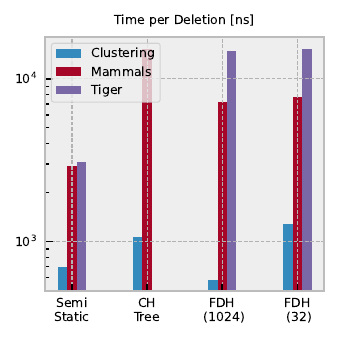}
    \caption{Comparing update times in nanoseconds of our algorithms on real-world data.}
    \label{fig:realworld:update}
\end{figure}

\subparagraph{Comparing query times.}
When comparing query times, we see a different image (Figure~\ref{fig:realworld:query}).
\competitorSimpleBTree dominates in performance.
This is to be expected: it statically stores the convex hull in a vector offering excellent cache locality and minimal overhead. 
We therefore conclude that between its good update performance and excellent query performance, a naive non-dynamic algorithm is the preferred method for moderately-sized real-world data sets.
If we compare \competitorFDH to \competitorCHTree, we observe that our queries are worse. This matches the theoretical fact that we have an asymptotic $O(\log^2 n)$ time query versus the $O(\log n)$ time query of a \competitorCHTree. 
We argue that comparing the static query performance of dynamic algorithms is of limited value and that our real-world data is of moderate size.  Instead, we should evaluate these algorithms on large-scale data sets and we should evaluate their overall performance under mixed sequences of updates and queries. For synthetic data, where we can ensure all algorithms do not crash, we can provide this comparison.

\begin{figure}[t]
    \centering
      \includegraphics[]{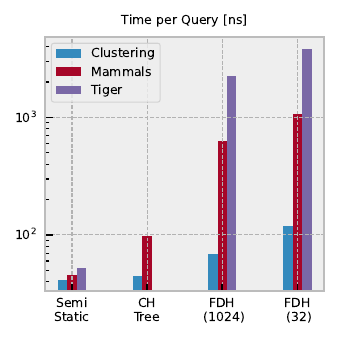}
    \includegraphics[]{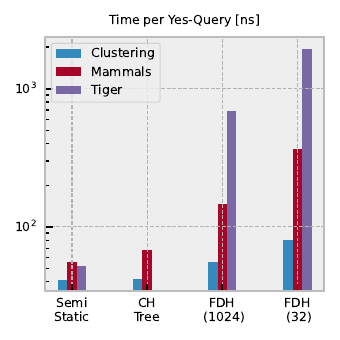}
    \caption{Comparing update times in nanoseconds of our algorithms on real-world data.}
    \label{fig:realworld:query}
\end{figure}

\subsection{Synthetic Experiments}\label{sec:synth}
Our synthetic data is drawn from four distributions, sampling points from a \textbf{Box}, a \textbf{Bell} curve, a \textbf{Disk}, or a \textbf{Circle}.
Our sampling procedure allows us to generate data sets of arbitrary size.
Moreover, since the sampled point sets contain very few points with shared coordinates, none of the competing algorithms crash on this data, and the rate of incorrect query answers is very low (but not zero).
This enables a fair comparison between the algorithms.
Recall that for these synthetic experiments that do not investigate scaling behaviour: we use point sets of between one and two million points.
For \competitorSimpleBTree, we have to omit the plots for the \textbf{Circle} distribution, as it times out after one hour.

\subparagraph{Overall trends.}
Figure~\ref{fig:sequences:overall} shows the running times for deletions and queries, as insertions and yes-queries exhibit similar trends.
In terms of update time, the \competitorSimpleBTree algorithm dominates performance on the \textbf{Bell} and \textbf{Box} data sets.
Both distributions have, in expectation, a convex hull of size $O(\log n)$, which is very sparse—indeed, even sparser than what we observed on the real-world data.
As a consequence, the algorithm performs very little work during an update: most deletions remove points that are not on the hull, and most insertions place points strictly inside the hull.
\competitorSimpleBTree also provides very strong query performance, and we again conclude that a naive, essentially non-dynamic, algorithm is preferred for application domains in which the convex hull is expected to be very~sparse.

\begin{figure}[h]
    \centering
    \includegraphics[]{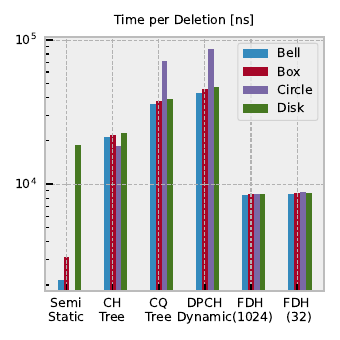}
          \includegraphics[]{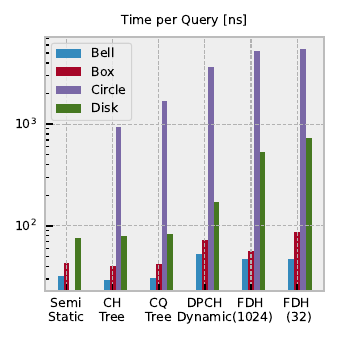}
        \caption{Time in nanoseconds per deletion (left) or queries (right) 
    on all four synthetic data sets. \vspace{-0.7cm}}
     \label{fig:sequences:overall}
\end{figure}

The \textbf{Disk} data is arguably our most realistic synthetic data set.
Here, we expect $O(n^{1/3})$ points to lie on the convex hull.
This aligns with our observations on real-world data, where arbitrary subsets consistently exhibit $O(n^{\varepsilon})$ points on the convex hull for some domain-dependent constant~$\varepsilon$.
Given the larger number of hull points, \competitorSimpleBTree begins to suffer in update performance.
It nevertheless remains the preferred method when compared to the published state of the art.
Indeed, we are somewhat surprised that these existing algorithms are outperformed by a naive static rebuilding strategy.
In contrast, our new technique supports updates approximately three times faster than this naive method.

When comparing \competitorFDH to all non-naive algorithms, we observe that our approach offers excellent update performance that is almost entirely agnostic to the input distribution.
This behaviour is consistent with our theoretical analysis: the tree-based solutions rely on algorithmic logic that depends strongly on the structure of the convex hull.
Our technique, by contrast, is primarily influenced by the \emph{number} of updates—since these determine when buckets are merged—rather than by the specific points involved.
In terms of query performance, our approach is approximately a factor of six slower on the \textbf{Disk} data.
On the \textbf{Circle} data, which has a very dense convex hull, the gap in update times widens while the gap in query times narrows slightly.
We claim that, taken together, the results on both the \textbf{Disk} and \textbf{Circle} distributions suggest that, under mixed update-query workloads, our method becomes the preferred technique compared to all alternatives, including \competitorSimpleBTree:
\vspace{-0.2cm}

    {%

\subparagraph{Mixing updates and queries.}
Our final point of comparison in the main body evaluates all algorithms on the \textbf{Disk} and \textbf{Circle} synthetic data sets. We refer to Appendix~\ref{sec:scaling} for our scaling experiments. 
As before, we fix a dynamic update sequence in which the point set size fluctuates between one and two million points.
We then use this data to consider ratios of insertions to deletions to queries of the form $1\!:\!1\!:\!x$, where $x$ ranges from $1$ to $25$ (see Figure~\ref{fig:sequences:mixing}).
We exclude the \textbf{Box} and \textbf{Bell} data sets from this comparison, as we have already established that \competitorSimpleBTree dominates in these settings.

Our first observation is that \competitorFDH is the preferred algorithm for ratios of $1\!:\!1\!:\!1$ and $1\!:\!1\!:\!2$, the latter corresponding to an equal number of updates and queries.
Both \competitorDPCH and \competitorCQTree are non-competitive across all ratios.
On the \textbf{Disk} data, both \competitorSimpleBTree and \competitorCHTree achieve reasonable performance.
For these methods, updates dominate the running time, and the proportion of queries must become very large before the overall running time decreases noticeably.
In contrast, for \competitorFDH the running time is dominated by queries, and thus increases sharply as the query ratio grows.
Nevertheless, \competitorFDH remains the best up to a ratio of $1\!:\!1\!:\!20$.

\begin{figure}[h]
    \centering
    \hspace{1.4cm}\includegraphics{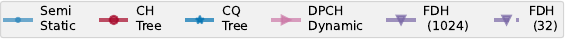}
    \includegraphics[]{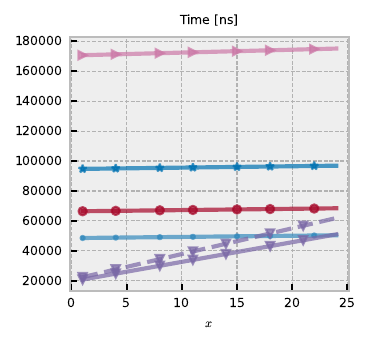}
          \includegraphics[]{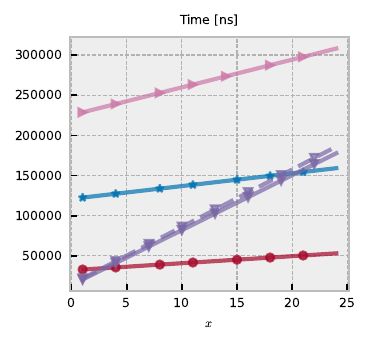}
    \caption{We plot the mean time in nanoseconds as we increase the ratio insertions:deletions:queries as $1:1:x$. We show the outcomes for the \textbf{Disk} data (left) and \textbf{Circle} data (right) \vspace{-0.6cm}}
    \label{fig:sequences:mixing}
\end{figure}

On the \textbf{Circle} data, \competitorSimpleBTree times out.
In this setting, every update triggers a linear-time rebuild of the convex hull, which is prohibitively expensive.
While \competitorFDH incurs significantly more expensive queries on this data set—primarily because many queries involve points that do \emph{not} lie inside the convex hull—our method must then compute tangents to each convex hull maintained in the data structure, which is costly.
Despite this data being  arguably unrealistic and adversarial to our technique, \competitorFDH remains the preferred approach for mixed workloads, even when updates and queries occur in equal proportion (i.e., $1\!:\!1\!:\!2$).

\section{Conclusion}

We present a new fully dynamic algorithm for supporting convex hull queries.
Our approach combines the logarithmic method with a deletion-only convex hull data structure, yielding amortised update times of $O(\log n \log \log n)$ and query times of $O(\log^2 n)$.
In contrast to existing techniques, our method deliberately trades slower queries for significantly faster updates.
We provide a non-trivial implementation that supports point-location queries, a particularly challenging and non-decomposable type of convex hull query.

We compare our implementation both to the state of the art and to a naive \competitorSimpleBTree baseline that rebuilds the convex hull whenever an update affects it.
Experiments on real-world data sets reveal that existing state-of-the-art techniques, in contrast to our implementation, lack robustness to duplicate coordinates.
These experiments further show that, for data sets with very sparse convex hulls, the naive \competitorSimpleBTree approach is the method of choice.

We argue, however, that realistic size-$n$ inputs often have convex hulls of size $O(n^{\varepsilon})$.
On such denser convex hulls, our method becomes the preferred approach for fully dynamic workloads that mix updates and queries.
In particular, update-heavy workloads strongly benefit from our technique, which is in line with its theoretical analysis.
On \textbf{Disk} data, our method remains preferable even when the ratio of updates to queries drops to $1\!:\!10$.
On the adversarial \textbf{Circle} data, where all points lie on the convex hull, performance is less decisive; nevertheless, down to a $1\!:\!1$ update--query ratio, our method still outperforms all competitors.
Our scaling experiments show that these trends persist as the workload increases.

Overall, we offer new theoretical insights and  extensive empirical validation.
We further contribute an openly available, robust implementation of fully dynamic convex hulls.

\bibliographystyle{plainurl}
\bibliography{references}
\clearpage
\appendix

\section{Scaling Experiments}\label{sec:scaling}
In this section we review the performance of the algorithms for increasing problem sizes across our \textbf{Box}, \textbf{Bell}, \textbf{Disk} and \textbf{Circle} synthetic data.
Our findings are best summarised by Figure~\ref{fig:full_data_scaling} which plots the median time taken on each data set on a 1:1:1 ratio of insertions, deletions and queries (for scaling input sizes $n$ up to $n = 2^{24}$

These plots are consistent with our findings in the main body: on the two \textbf{Box} and \textbf{Bell} data sets that have very sparse convex hulls, the \competitorSimpleBTree method dominates in performance. 
On data sets with denser hulls it becomes non-competitive and on \textbf{Circle} it even times out. Our \competitorFDH method is second-best on sparse data and the preferred method on dense convex hulls. 
Our scaling experiments show that this remains true as the input size increases and we would argue that the scaling behaviour of these algorithms is roughly identical. The exception is the  \competitorSimpleBTree algorithm. This algorithm has a scaling that is linear in the expected complexity of the convex hull and so its scaling behaviour is domain-dependent. This is illustrated on the \textbf{Disk} data where its running time deteriorates significantly faster than all other algorithms.

\begin{figure}[H]
    \centering
    \hspace{1cm}\includegraphics{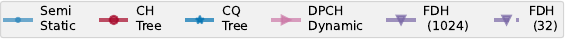}
    \includegraphics{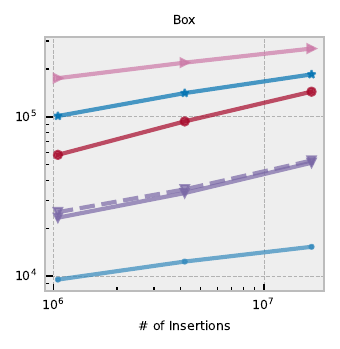}   
    \includegraphics{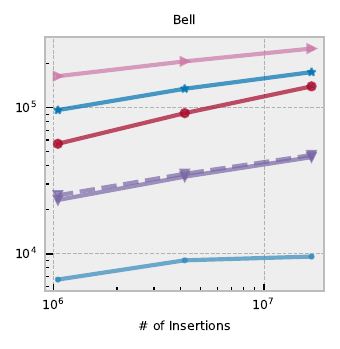}
    \includegraphics{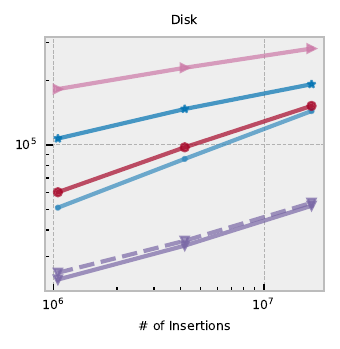}
    \includegraphics{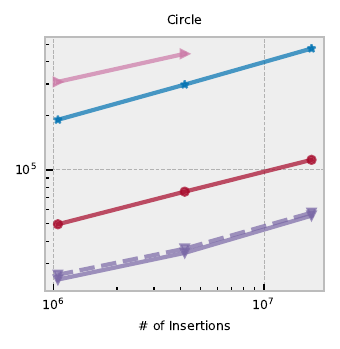}
    \caption{We plot the median time taken by an operation, where operations have a 1:1:1 ratio between insertions, deletions and queries. 
    We show this plot for each of our four synthetic data generators as the input sizes increase. }
    \label{fig:full_data_scaling}
\end{figure}

\subparagraph{Bell Data.} 
 In Figure~\ref{fig:scaling:bell} we present the results for \textbf{Bell} generated data.
 In general, \competitorDPCH is slower than all other methods for all tasks. Our new \competitorFDH methods can match the insertion time for $n=2^{24}$ of \competitorSimpleBTree.
For deletions the \competitorSimpleBTree algorithm has a near constant time per deletion, being up to 8.6 times faster than the next fastest,  \competitorFDH.
All queries in this instance are yes queries, and our \competitorFDH methods are faster than \competitorDPCH. Interestingly, the \competitorCHTree and \competitorCQTree have a slight edge over \competitorSimpleBTree.

\begin{figure}
    \centering
    \hspace{1.4cm}\includegraphics{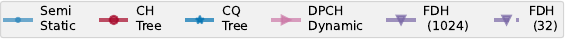}\newline
    \includegraphics[]{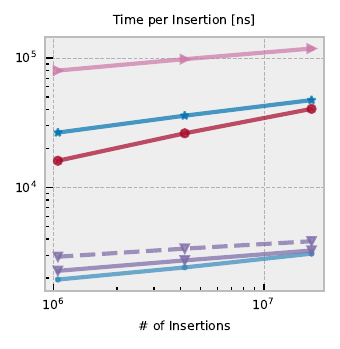}
    \includegraphics[]{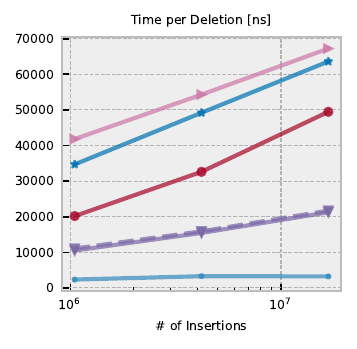}
    \includegraphics[]{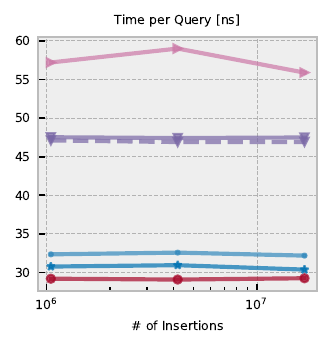}
    \includegraphics[]{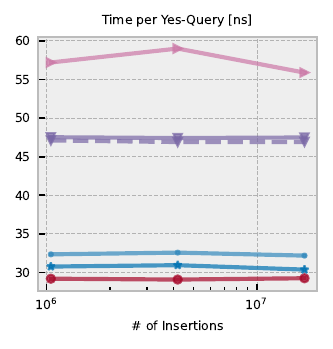}
    
    \caption{Update and query times in nanoseconds for the \textbf{Bell} data generator.}
    \label{fig:scaling:bell}
\end{figure}
\subparagraph{Box Data.} Figure~\ref{fig:scaling:box} shows the results for \textbf{Box} generated data. For insertions, our \competitorFDH method is faster than \competitorSimpleBTree.
When the first-bucket is bigger, the effect is more pronounced. The \competitorCHTree and \competitorCQTree are more than one order of magnitude slower, while, \competitorDPCH is nearly two orders slower.
When deleting, the small hulls for \textbf{Box} data is not affected that often. The \competitorSimpleBTree is the only algorithm not requiring work on a non-changing deletion.
Therefore, it is the most efficient method for deleting.
All queries are yes-instances for this generator.
The time per query is highly dependent on the size of the first bucket for the \competitorFDH algorithm. With a bigger bucket, the overall number of buckets is lower, resulting in up to 40\% faster queries.
Naturally, the \competitorCHTree,\competitorCQTree and \competitorSimpleBTree offer the fastest query responses,being about a third faster than \competitorFDH.
\begin{figure}
    \centering
    \hspace{1.4cm}\includegraphics{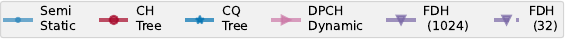}\newline
    \includegraphics[]{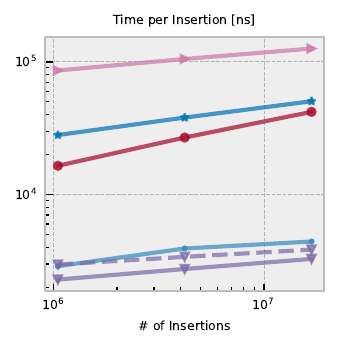}
    \includegraphics[]{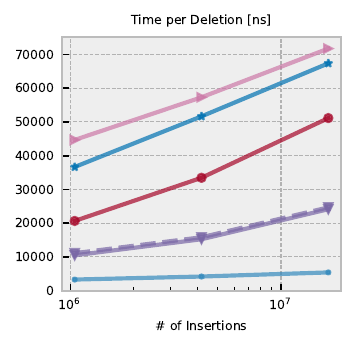}
    \includegraphics[]{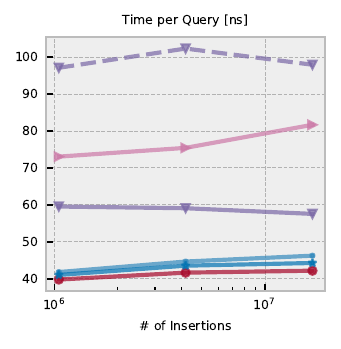}
    \includegraphics[]{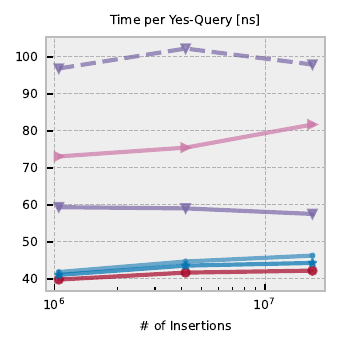}
        \caption{Update and query times in nanoseconds for the \textbf{Box} data generator.}
    \label{fig:scaling:box}
\end{figure}

\subparagraph{Disk Data.}
Figure~\ref{fig:scaling:disk} shows the result for data following a disk distribution.
Since $n^{1/3}$ points are on the hull, the naive \competitorSimpleBTree struggles by having to recompute the hull repeatedly. The \competitorFDH method can take advantage of its logarithmic structure and is the fastest for insertions and deletions by a wide margin.
The logarithmic structure is also the reason, why we see higher response times for queries.
The \competitorFDH with a first bucket size of 1024 is st least 2.25 times faster per yes-query than its size 32 implementation. For queries overall, this trend is also observable.
The other methods, maintaining only one hull, answer the queries faster, yes-queries at least 25\%.

\begin{figure}
    \centering
    \hspace{1.4cm}\includegraphics{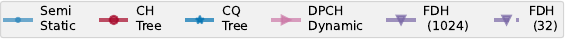}\newline
    \includegraphics[]{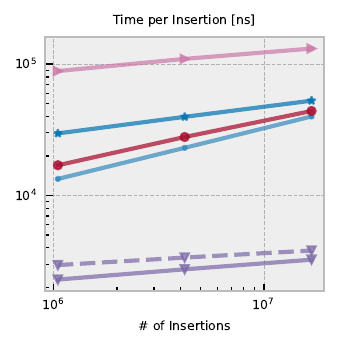}
    \includegraphics[]{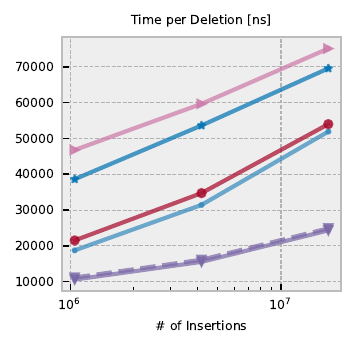}
    \includegraphics[]{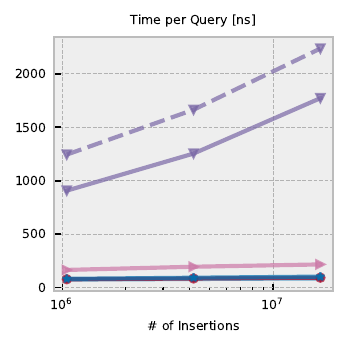}
    \includegraphics[]{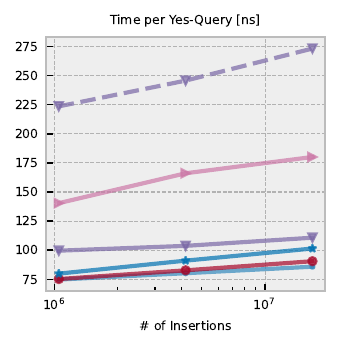}
        \caption{Update and query times in nanoseconds for the \textbf{Disk} data generator.}
    \label{fig:scaling:disk}
\end{figure}
\subparagraph{Circle Data.} Figure~\ref{fig:scaling:circle} contains the result for the extreme case, where all points lie on a circle and their hull. Naturally, \competitorSimpleBTree times out on every instance and is not present in the plots. \competitorDPCH exceeds the one hour runtime cut-off for $n=2^{24}$.
 \competitorFDH is one order of magnitude faster for insertions than the second fastest, \competitorCHTree. \competitorFDH also manages deletions the fastest. For queries we can observe two trends.
Firstly, answering yes-queries is faster with our new method, while overall queries (skewed by no-queries) are  costly with our new method. 
The \competitorFDH has to collect the expanding points for a no-query point in every $\log{n}$ sub-hull and compute the outer hull.

\begin{figure}
    \centering
    \hspace{1.4cm}\includegraphics{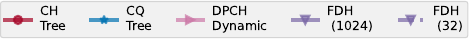}\newline
    \includegraphics[]{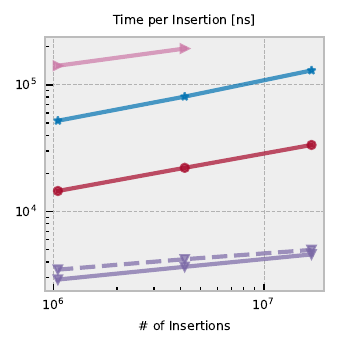}
    \includegraphics[]{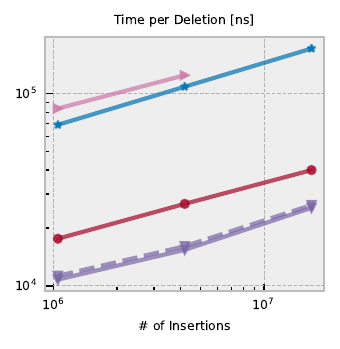}
    \includegraphics[]{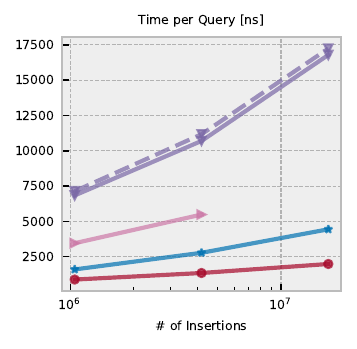}
    \includegraphics[]{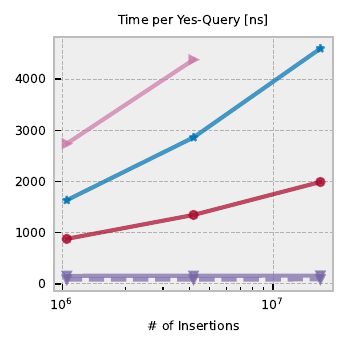}
    
    \caption{Update and query times in nanoseconds for the \textbf{Circle} data generator.}    \label{fig:scaling:circle}
\end{figure}

\end{document}